\documentclass[journal]{IEEEtran}

\IEEEoverridecommandlockouts                              

\usepackage{graphicx}
\usepackage{subcaption} 
\usepackage{graphics} 
\usepackage{epsfig} 

\usepackage{amsmath} 
\usepackage{amssymb}  
\usepackage{booktabs} 
\usepackage{mathtools}

\usepackage{amsthm}
\usepackage{multirow}
\usepackage{pgfplots}
\usepackage{tikz}
\pgfplotsset{compat=1.18} 
\usetikzlibrary{arrows.meta, positioning} 

\usepackage{flushend}
\usepackage[noadjust]{cite}
\usepackage{xcolor}
\usepackage[capitalize]{cleveref}

\definecolor{darkgreen}{rgb}{0.0, 0.5, 0.0} 

\newtheorem{remark}{Remark}
\newtheorem{theorem}{Theorem}

\newtheorem{corollary}{Corollary}
\newtheorem{definition}{Definition}

\begin{document}
\title{Exact Characterization of Aggregate Flexibility via Generalized Polymatroids}

\author{
    Karan Mukhi,
    Georg Loho,
    and Alessandro Abate%
    \thanks{K. Mukhi and A. Abate are with the Department of Computer Science, University of Oxford, UK}
    \thanks{G. Loho is with the Department of Applied Mathematics, University of Twente, The Netherlands and the Department of Mathematics and Computer Science, FU Berlin, Germany}
}


\maketitle

\begin{abstract}
It is well established that the aggregate flexibility inherent in populations of distributed energy resources (DERs) can be leveraged to mitigate the intermittency and uncertainty associated with renewable generation, while also providing ancillary grid services.
To enable this, aggregators must effectively represent the flexibility in the populations they control to the market or system operator. A key challenge is accurately computing the aggregate flexibility of a population, which can be formally expressed as the Minkowski sum of a collection of polytopes, a problem that is generally computationally intractable.
However, the flexibility polytopes of many DERs exhibit structural symmetries that can be exploited for computational efficiency. To this end, we introduce \textit{generalized polymatroids}, a family of polytopes, into the flexibility aggregation literature.
We demonstrate that individual flexibility sets belong to this family, enabling efficient computation of their \textit{exact} Minkowski sum. 
For homogeneous populations of DERs we further derive simplifications that yield more succinct representations of aggregate flexibility.
Additionally, we develop an efficient optimization framework over these sets and propose a vertex-based disaggregation method, to allocate aggregate flexibility among individual DERs.
Finally, we validate the optimality and computational efficiency of our 
approach through comparisons with existing methods.
\end{abstract}

\section{Introduction}

Power systems are increasingly integrating renewable energy resources to reduce environmental impact.
However, the intermittency and uncertainty of generation from wind and solar poses significant challenges. As the proportion of renewables rises, power system operators must adopt more flexible operational strategies to address these challenges. Meanwhile, with the electrification of heating and transportation, electricity demand is rising significantly \cite{Baruah2014EnergyElectrification}, and uncontrolled consumption from these new loads can strain the grid \cite{Lopes2011IntegrationSystem}. Many of these devices are inherently flexible, meaning there exist a set of consumption profiles that respect the operational constraints of the devices. For example, electric vehicles (EVs) are typically plugged in far longer than is necessary for a complete charge \cite{Lee2019ACN-Data:Dataset}, and thermostatically controlled loads (TCLs) can operate within a dead-band around their set-point temperature \cite{Callaway2009TappingEnergy}.
Distributed energy resources (DERs) encompass a diverse range of such small-scale loads and generators whose energy consumption can be actively managed.
Collectively, ensembles of these devices may offer substantial flexibility that can be used to mitigate the variability and uncertainty associated with renewable energy generation \cite{Almassalkhi2023IntelligentDecarbonization}.

Various control architectures have been proposed to make use of this flexibility, these can generally be categorized into centralized, decentralized and hierarchical schemes \cite{Callaway2011AchievingLoads}.
Fully centralized schemes, where devices report constraints to a system operator that solves a monolithic optimization and issues commands, are globally optimal but impractical and poorly scalable for large DER populations.
Decentralized approaches, such as those proposed in \cite{Tindemans2015DecentralizedResponse} and \cite{Gan2012OptimalCharging}, mitigate scalability challenges. However, they necessitate local computational capabilities in devices and may be unsuitable for real-time decision-making due to latency constraints. More seriously, the lack of a centralized decision-making entity makes it hard to assign accountability for maintaining strict grid reliability guarantees. Hierarchical control schemes can alleviate these issues, whereby an aggregator controls a population of DERs and presents its aggregate flexibility to the system operator \cite{Xu2016HierarchicalLoads}. In this scheme, the aggregator becomes the accountable party that can provide the reliability guarantees to system operators, whilst reducing the complexity for the system operator. An integral part of the function of an aggregator is to represent the aggregate flexibility in the population it controls.  

\subsection{Related Work}
Continuous-time exact representations of flexibility for energy storage systems (ESS) are proposed in \cite{Evans2020AResources, Zachary2021SchedulingStorage}. However,  these representations are ill-suited for the typical discrete-time operation of power systems. Indeed, most work considers the discrete-time variant of the problem, modeling consumption profiles as piecewise-constant functions. In this paradigm, the flexibility of a variety of DERs can be represented as a convex polytope that encode all operational constraints of the device \cite{Trangbaek2011ExactControl, Zhao2017ALoads, Hao2015AggregateLoads}. In some cases, the flexibility sets may be non-convex \cite{Taha2024WhenConvex}, however we will focus on consumption models that generate convex sets in this paper.
The aggregate flexibility of the population can then be characterized exactly as the Minkowski sum of these individual flexibility polytopes. In general, computing the Minkowski sum is difficult \cite{Tiwary2008OnPolytopes}, so most of the literature focuses on approximating the Minkowski sum. Efficient methods of computing outer approximations are proposed in \cite{Barot2017APolytopes, Wen2022AggregateModels}, however by definition these over-approximations include infeasible aggregate consumption profiles and therefore lack the reliability guarantees that are essential in power systems. Accordingly, most of the literature focuses on computing inner approximations of the sum:  many results focus on finding inner approximations of each flexibility set with a base polytope. These base polytopes are selected so that computing the Minkowski sum of a set of them is efficient, \cite{Kundu2018ApproximatingApproach, Muller2019AggregationResources,Alizadeh2014CapturingResponse, Taha2024AnPopulations}. In \cite{Muller2019AggregationResources} zonotopes are used,
whilst in \cite{Zhao2017ALoads} the authors describe a method of defining a base polytope finding inner approximations of the individual flexibility sets with homothets of this base. In \cite{Taha2024AnPopulations} this approach is generalized so that inner approximations may be affine transformations of the base set. Another scheme for deriving inner approximations is by viewing the Minkowski sum as a projection, and computing an inner approximation in the pre-projection space \cite{Zhao2016ExtractingApproximation}. \cite{Nazir2018InnerResources} forms an inner approximation as the union of a set of boxes: this can yield an arbitrarily good approximation for the polytope at the cost of increasing the computational burden. Inner approximations do not contain all feasible aggregate consumption profiles and so one cannot guarantee optimality when optimizing over them, indeed some of the approximations can be very conservative \cite{Ozturk2022AggregationAlgorithms}.
Clearly, only computing exact characterizations can guarantee optimality and feasibility, to this end \cite{Panda2024EfficientVehicles} and \cite{Mukhi2023AnVehicles} focus on a specific, but relevant, case of populations of EVs with charging only capability. These works characterize the individual flexibility sets as permutahedra and use properties of this class of polytopes to perform the Minkowski sum efficiently. The aggregate flexibility sets derived in \cite{Wen2022AggregateModels}, \cite{Mukhi2023AnVehicles}
 and \cite{Panda2024EfficientVehicles} are specific instantiations of the family of polytopes that we study in this paper, and the results in those works can be recovered from the more general theory presented here. 

 Finally, stochastic variants of the problem, where operational constraints of the individual devices are uncertain, have also been studied in \cite{Taheri2022Data-DrivenModels, Zhang2024AUncertainty}. Incorporating uncertainty is beyond the scope of this paper, however results from this paper can be applied to this extended setting \cite{Mukhi2025RobustFlexibility}. 
\subsection{Main Contributions}
 In this context, the contributions of this paper are summarized as follows:
\begin{itemize}
    \item We use generalized polymatroids (g-polymatroids) as a base polytope, and show that the flexibility set of a broad class of DERs can be represented exactly as a g-polymatroids. 
    \item Leveraging properties of g-polymatroids, we derive exact representations of the aggregate flexibility set for a population of DERs.
    \item For a population of  charging-only EVs, we show how these representations may be simplified further.
    \item Applying tools from combinatorial optimization we provide efficient methods of optimizing over these sets.
    \item Finally, we propose a vertex-based method to disaggregate an aggregate consumption profile among devices in the population.
\end{itemize}

The rest of this paper is structured as follows: in \cref{sec:prob_form} we introduce our DER consumption models and formalize the aggregation problem. \cref{sec:aggregation} begins with a brief overview of g-polymatroids and then proceeds to show how individual and aggregate flexibility sets may be represented as g-polymatroids. This is followed with a discussion on simplified representations for homogeneous populations.
In \cref{sec:optimization} we discuss how these sets may be efficiently optimized over, and \cref{sec:disaggregation} provides a vertex based disaggregation method. Numerical studies are presented in \cref{sec:numerical_results}, benchmarking the complexity and optimality of this work against competing methods. Finally, conclusions are drawn in \cref{sec:conc}.

 \subsection*{Notation}\label{subsection:notation}
 \noindent
 For a vector $ u \in \mathbb{R}^{\mathcal{T}} $, where $ \mathcal{T} \subseteq \mathbb{N} $ is a finite index set and $ t \in \mathcal{T} $, we denote by $ u(t) $ the $ t^{\text{th}} $ component of $ u $.
 For any subset $ A \subseteq \mathcal{T} $,  we let $A' := \mathcal{T}\setminus A$.
 We also define $ u(A) := \sum_{t \in A} u(t) $.
 In particular, for $ t \in \mathcal{T} $, we write $ [t] := \{1, 2, \ldots, t\} $.
 Lastly, we use the notation $ \sum(\cdot) $ to denote both standard summation and Minkowski summation, with the specific meaning determined by context.

\section{Problem Formulation}\label{sec:prob_form}
In this section we introduce our power consumption model for a DER and formalise our notion of \textit{flexibility}, both in the context of individual devices and for aggregations of devices.
We introduce three problems that we aim to solve in this paper: aggregation, optimisation and disaggregation. 
Finally, we discuss the expressivity and the limitations of the model we introduce.

In the following, we consider an aggregator that has direct control over the power consumption of a finite population of devices, indexed by $i \in \mathcal{N}$, where $\mathcal{N} := \{1, \ldots, N\}$. 
We consider this problem over a finite time horizon, we discretize this horizon into $T$ time steps each of duration $\delta$. We denote the set of time steps as $\mathcal{T} := \{1, \ldots, T\}$, and let $t \in \mathcal{T}$ index a specific interval. 

\subsection{DER Flexibility Sets}\label{subsect:der_model}
Let $u_i(t)$ denote the consumption rate of the $i^{th}$ DER, in time step $t \in \mathcal{T}$. Each time step is assumed to be of equal length, denoted $\delta > 0$.  By convention, $u_i(t)$ denotes the net power consumption of the device, i.e. $u_i(t) < 0$ indicates that the device is generating power, and $u_i(t) > 0$ signifies the device is consuming power. The DER's power consumption is assumed to be constant within each time step. The vector $u_i \in \mathbb{R}^{\mathcal{T}}$ denotes the \textit{consumption profile} of the DER over the entire time horizon.

Each DER will have a (possibly time-dependent) lower and upper limit on its power consumption, denoted by $\underline{u}_i$ and $\overline{u}_i \in \mathbb{R}^\mathcal{T},$ such that its power consumption must stay within this interval in each time step:
\begin{equation}
    \underline{u}_i(t) \leq u_i(t) \leq \overline{u}_i(t) \quad \forall t \in \mathcal{T}.
\end{equation}
Next, we let $x_i \in \mathbb{R}^{\mathcal{T}}$ denote the \textit{state of charge} (SoC) of the DER, such that $x_i(t)$ is the state of charge at the end of timestep $t$:
\begin{equation}\label{eq:soc_dynamics}
    x_i(t) = x_i(t-1) + u_i(t) \delta.
\end{equation}
Without loss of generality, we assume $\delta=1$. By convention, and also without loss of generality, we assume the initial state of charge of the DER at $t=0$ is $x_i(0) = 0$, and so we can write \eqref{eq:soc_dynamics} as:
\begin{equation}
    x_i(t) = u_i([t])
\end{equation}
where $u_i([t])$ denotes the sum of the first $t$ elements of $u_i$, as introduced in the notation section.
Similarly to its power constraints, the DER will also have a (again possibly time-dependent) lower and upper limit on the SoC, denoted by $\underline{x}_i$ and $\overline{x}_i \in \mathbb{R}^{\mathcal{T}}$ respectively, such that:
\begin{equation}\label{eq:energy_constraints}
    \underline{x}_i(t) \leq u_i([t]) \leq \overline{x}_i(t) \quad \forall t \in \mathcal{T}.
\end{equation}
For ease of notation we collect all the parameters relating to the DER power consumption requirements into the tuple 
$\xi_i = (\underline{u}_i, \overline{u}_i, \underline{x}_i, \overline{x}_i)$.

\begin{definition}\label{dfn:individual_flexibility_sets}
    For a DER with consumption requirements $\xi_i$, the \emph{individual flexibility set} of the device, denoted $\mathcal{F}(\xi_i)$, is the set of all feasible consumption profiles for the DER:
\begin{equation*}
    \mathcal{F}(\xi_i) := \left\{ u \in \mathbb{R}^{\mathcal{T}} \; \middle\vert \;
   \begin{array}{@{}cl}
                    \underline{u}_i(t) \leq u(t)  \leq \overline{u}_i(t) \;\; \forall t \in \mathcal{T}\\
                    \underline{x}_i(t) \leq \sum_{s=1}^t u(s) \leq \overline{x}_i(t) \;\; \forall t \in \mathcal{T}
   \end{array} 
   \right\}. 
\end{equation*}
\end{definition}
The individual flexibility set, $\mathcal{F}(\xi_i)$, is defined by a set of linear constraints.  The first set of these constraints is clearly bounded, hence the individual flexibility sets are all compact, convex polytopes.

\subsection{Aggregated Flexibility Sets}
We now consider an aggregator controlling a population of DERs, each characterized by its own consumption parameter. We let $\Xi_\mathcal{N} = \{\xi_i\}_{i \in \mathcal{N}}$ denote the multiset of consumption parameters for all DERs in the population. 
The \textit{aggregate consumption profile}, denoted $u_\mathcal{N}$, is the sum of the individual consumption profiles of devices in the population:
\begin{equation}
    u_\mathcal{N} = \sum_{i \in \mathcal{N}} u_i.
\end{equation}
\begin{definition}\label{dfn:aggregate_flexibility_set}
    The \emph{aggregate flexibility set} of a population of DERs  with consumption parameters $\Xi_\mathcal{N}$ is the set of all feasible aggregate consumption profiles of the population:
\end{definition}
\begin{equation*}
    \mathcal{F}(\Xi_\mathcal{N}) := \left\{ u_\mathcal{N} \in \mathbb{R}^{\mathcal{T}} \middle| u_\mathcal{N} = \sum_{i \in \mathcal{N}} u_i,\;\; u_i \in \mathcal{F}(\xi_i) \; \forall i \in \mathcal{N} \right\}. 
\end{equation*}
By definition, the aggregate flexibility set $ \mathcal{F}(\Xi_\mathcal{N})$, is the Minkowski sum of the individual flexibility sets of the DERs in the population:
\begin{equation}
    \mathcal{F}(\Xi_\mathcal{N})  = \sum_{i \in \mathcal{N}} \mathcal{F}(\xi_i).
\end{equation} 

\begin{remark}
One can compute the Minkowski sum of a collection of polytopes by considering the common refinement of their normal fans \cite[Proposition 7.12]{Ziegler2012LecturesPolytopes}.
By summing each of the polytope’s support functions of the rays of the refined normal fan, one obtains the facet description of the sum, and by summing the support functions of the full-dimensional cones, one recovers the vertex description.
Forming the common refinement of several normal fans can become quite expensive, especially as the dimension or the number of summands grows. In essence, one must take each cone from every input fan and intersect it with all cones coming from the other fans. In practical terms, this means that in higher dimensions—or when combining many polytopes—enumerating all intersections tends to blow up combinatorially. Because our aim is to compute the Minkowski sum of a large family of such high-dimensional polytopes, forming the common refinement is generally infeasible. However, we will show that the normal fans of the flexibility sets in question are all coarsenings of a single fan, so the normal fan of their Minkowski sum is exactly that shared fan.
\end{remark}

\subsection{Optimization and Disaggregation}
With a characterization of the aggregate flexibility set $\mathcal{F}(\Xi_\mathcal{N}) $, we consider aggregators that would like to find certain optimal aggregate consumption profiles, solving problems of the form:
\begin{equation}
    \textrm{minimize} \;\; f(u) \quad s.t \;\;u \in \mathcal{F}(\Xi_{\mathcal{N}}).
\end{equation}
As we shall see in the next section, the characterization of the aggregate flexibility set, though exact, is complex. In general, $\mathcal{F}(\Xi_\mathcal{N}) $ is a polytope characterized by at most $2^T$ facets and up to $T!$ vertices. Optimizing over this characterization is not trivial, and so we shall present methods to make this optimization tractable.

Finally, given an optimal aggregate consumption profile $u_\mathcal{N}^* \in \mathcal{F}(\Xi_\mathcal{N})  $, we seek to disaggregate this profile among the DERs in the population. This involves determining individual consumption profiles that are feasible for each device while ensuring their aggregate consumption matches the optimal aggregate consumption profile, i.e. computing $u_i^*$ such that:
\begin{equation}\label{eq:disaggregation}
  u_\mathcal{N}^* = \sum_{i=1}^N u_i^* \quad s.t \;\; u_i^* \in \; \mathcal{F}(\xi_i) \;\; \forall i \in \mathcal{N}.
\end{equation}

\subsection{Expressivity of the Model}\label{subsect:expressivity}
The DER consumption model introduced in Section \ref{subsect:der_model} is versatile, capable of representing the flexibility in a variety of devices, including EVs \cite{Taha2024AnPopulations}, storage systems, distributed generation and the slower dynamics of TCL consumption \cite{Xu2016HierarchicalLoads}. This section details the methodology for specifying the model parameters, denoted by $\xi_i$. 
For brevity, the discussion is limited to EVs; however, the approach can be readily extended to other classes of devices.

We assume that each vehicle arrives at the charging station at the beginning of time step $a_i$ and departs at the end of time step $d_i$, where $a_i, d_i \in \mathcal{T}$,  within the defined time horizon. Accordingly, we define the charging interval for vehicle $i$ as $C_i := \{a_i, a_i +1,...,d_i\} \subseteq \mathcal{T}$
At all time steps during this period, the EV may consume energy between its minimum and maximum power capacity $\underline{m}_i$ and $\overline{m}_i$, whereas for all other time steps the power consumption must vanish. Accordingly, we establish the values of $\underline{u}_i(t)$ and $\overline{u}_i(t)$:
\begin{equation*}
    \begin{array}{cc}
        \underline{u}_i(t) = 
        \begin{cases}
            0                   & t \notin C_i  \\
            \underline{m}_i     & t \in C_i,\\
        \end{cases}
        & 
        \overline{u}_i(t) = 
        \begin{cases}
            0                   & t \notin C_i \\
            \overline{m}_i     &  t \in C_i.\\
        \end{cases}
    \end{array}
\end{equation*}
This model permits discharging and so $\underline{m}_i$ may take negative values. To model EVs with no discharging capabilities we simply set $\underline{m}_i = 0$ \cite{Panda2024EfficientVehicles}, \cite{Mukhi2023AnVehicles}.
Next, we consider the constraints on the state of charge of the EV. We assume each EV has a limited energy storage capacity, $\overline{x}_i \in \mathbb{R}_+$. The EV arrives with an initial state of charge $e^0_{i}$, and must have a final state of charge in the interval $[\underline{e}_i, \overline{e}_i]$ at the time it departs, from which we determine the values for $\underline{x}_i(t)$ and $\overline{x}_i(t)$ \cite{Taha2024AnPopulations} \cite{Hao2014CharacterizingLoads} 
\begin{equation*}
    \begin{array}{cc}
        \underline{x}_i(t) = 
        \begin{cases}
            -e^0_{i}                  & t < d_i\\
            \underline{e}_i -e^0_{i}  & t \geq d_i,
        \end{cases} 
        & 
        \overline{x}_i(t) = 
        \begin{cases}
            \overline{x}_i -e^0_{i}   & t < d_i\\
            \overline{e}_i -e^0_{i}   & t \geq d_i.
        \end{cases} 
    \end{array}
\end{equation*}

An energy storage system is essentially an EV, that is available for the entire time horizon, i.e. $\underline{u}(t) = \underline{m}$ and  $\overline{u}(t) = \overline{m}$, $\forall t$. Furthermore, there are no constraints on the final state of charge other than respecting the device's storage constraints. Assuming the ESS has an initial SoC $e^0_{i}$, we have $\underline{x}_i(t) = -e^0_{i}$ and $\overline{x}_i(t) = \overline{x}_i - e^0_{i}$, $\forall t$.

For distributed generation the consumption will be bounded from above by $\overline{u}_i(t) = 0 \; \forall t$, i.e. when the device is fully curtailed, and from below by $\underline{u}_i(t)$, its time-dependent maximum power output. Generation systems lack a SoC, and so to effectively disregard SoC constraints, we set  $\underline{x}_i(t) = -\infty$ and $\overline{x}_i(t) = \infty \;\; \forall t$.

\subsection{Limitations}
Whilst this model is expressive enough to describe  various different classes of devices, there are some limitations. Firstly, this model assumes perfect round-trip-efficiency for battery charging and discharging. The state of charge dynamics for batteries from \eqref{eq:soc_dynamics} can more generally be written as
\begin{equation*}
    x_i(t) =  x_i(t-1) + \eta_i^+ u_i^+(t) + \frac{1}{\eta_i^-} u_i^-(t), 
\end{equation*}
where $\eta_i^+$ and $\eta_i^-$ denote the device's charging and discharging efficiencies, and $ u_i^+(t) := \textrm{max}\{0, u_i(t)\}$ and $u_i^-(t) := \textrm{min}\{0, u_i(t)\}$ separate the consumption profile into the charging and discharging components. Whilst our model implicitly assumes $\eta_i^+ = \eta_i^- = 1$, $\eta_i^+$ and $\eta_i^-$ will typically lie in the interval $(0, 1]$ depending on the characteristics of the battery. However, if the devices are restricted to consumption only, i.e. $\underline{u}_i(t) \geq 0$ for all $t$, we may absorb $\eta_i^+$ into $u_i^+(t)$,  and the model remains faithful.

Secondly, we assume lossless charging, i.e. devices perfectly hold their state of charge. The state dynamics from \eqref{eq:soc_dynamics} for general lossy devices can instead be written as 
\begin{equation}
    x_i(t) = \lambda_i x_i(t-1) + u_i(t), 
\end{equation}
with $\lambda_i \in [0,1]$, where $\lambda_i = 1$ corresponds to lossless charging. While this is generally not a concern when modeling storage devices, since battery based systems such as EVs and ESSs typically retain their charge effectively, the consumption dynamics of TCLs are inherently characterized by lossy storage. Nevertheless, g-polymatroid inner or outer approximations can still be constructed when this assumption is relaxed, as shown in \cite{Mukhi2025AggregatePolymatroids}.

Other methods in the literature do not impose these assumptions \cite{Taha2024AnPopulations, Muller2019AggregationResources}, however they yield approximations of the aggregate flexibility.
It is precisely the enforcement of these assumptions, namely, perfectly efficient charging and discharging ($\eta_i^+ = \eta_i^- = 1$) and lossless storage ($\lambda_i = 1$), that gives rise to the structure of the flexibility sets that allows them to be characterized as g-polymatroids. These structural properties are elaborated in Remark~\ref{rem:genpolyedges}.

\section{Aggregation}\label{sec:aggregation}
In this section we provide an exact characterisation of the aggregate flexibility sets defined in the previous section. Our strategy is to show that the individual flexibility sets are g-polymatroids, a family of polytopes that can be represented with a pair of super- and submodular functions. We are able to leverage properties of g-polymatroids to provide exact representations of the aggregate flexibility. In general, these functions will be complex, so we end this section with an example of how they may be simplified, and link this back to results in the literature. 

\subsection{Generalized Polymatroids}
For the sake of completeness we provide a brief overview of g-polymatroids. We refer the reader to \cite{Frank2011ConnectionsOptimization} for a more detailed treatment of the subject.

\begin{definition}\label{dfn:submodular}
    A \emph{submodular function} $b: 2^\mathcal{T} \rightarrow \mathbb{R}$ is a set function defined over subsets of a finite set $\mathcal{T}$ that satisfies the property:
    \begin{equation}\label{eq:submodular}
        b(A) + b(B) \geq b(A \cap B) + b(A \cup B)
    \end{equation}
    for all subsets $A, B \subseteq \mathcal{T}$.
\end{definition}
One can also define a \emph{supermodular function}, $p$, by reversing the inequality from \eqref{eq:submodular}, or as the negative of a submodular functions, i.e. $p = -b$ is supermodular if $b$ is submodular.

\begin{definition}
    The \emph{submodular polyhedron}, $ \mathcal{P}(b) \subseteq \mathbb{R}^{\mathcal{T}}$, associated with the submodular function $b$ is defined as:
\begin{equation*}
     \mathcal{P}(b) = \left\{ u \in  \mathbb{R}^{\mathcal{T}} \mid u(A) \leq b(A) \;\; \forall A \subseteq \mathcal{T}\right\}.
\end{equation*}
\end{definition}
Similarly we define a \emph{supermodular polyhedron} associated with the supermodular function $p$ as:
\begin{equation*}
     \mathcal{P}'(p) = \left\{ u \in  \mathbb{R}^{\mathcal{T}} \mid   p(A) \leq u(A)\;\; \forall A \subseteq \mathcal{T}\right\}.
\end{equation*}
Note how $ \mathcal{P}(b)$ and $ \mathcal{P}'(p)$ are defined by a set of $2^{|\mathcal{T}|}$ hyperplanes, one for each of the $2^{|\mathcal{T}|}$ subsets of $\mathcal{T}$.
 
\begin{definition}
    The pair $(p,b)$ is said to be \emph{paramodular} if $p$ is supermodular, $b$ is submodular, $p(\emptyset) = b(\emptyset) = 0$ and the \emph{cross-inequality}:
    \begin{equation}
        b(A) - p(B) \geq b(A\setminus B) - p(B \setminus A)
    \end{equation} 
    holds for all $A, B \subseteq \mathcal{T}$.
\end{definition}
The cross inequality is equivalent to ensuring the base polyhedron of $p$ is contained within the submodular polyhedron of $b$, and vice-versa.

\begin{definition}\label{dfn:g_polymatroid}
    For a paramodular pair $(p,b)$ we define the \emph{generalized polymatroid (g-polymatroid)}, denoted $\mathcal{Q}(p,b)$ as the polytope:
    \begin{equation}
        \mathcal{Q}(p,b) = \left\{ u \in  \mathbb{R}^{\mathcal{T}} \mid p(A) \leq u(A) \leq b(A) \;\; \forall A \subseteq \mathcal{T} \right\}.
    \end{equation}
\end{definition}
For any paramodular pair, the g-polymatroid defined by the pair is non-empty.
Essentially, a g-polymatroid is the intersection of the supermodular and submodular polyhedra associated with the paramodular pair $(p,b)$, as illustrated in Fig. \ref{fig:sub_super_polyhedra}. G-polymatroids comprise a rich class of polytopes that have been studied in the context of combinatorial optimization. This class of polytopes subsumes many other classes of common polytopes, including cubes, simplexes and permutahedra \cite{Postnikov2009PermutohedraBeyond}. 

\begin{figure}
    \centering
    \begin{tikzpicture}

    \filldraw[opacity=0.3, red] (1,3) -- (3,1) -- (6,1) -- (6,6) -- (1,6) -- cycle;
    \filldraw[opacity=0.3, blue] (0,0) -- (4.5,0) -- (4.5,3.5) -- (3.5,4.5) -- (0,4.5) -- cycle;
    \draw[dashed, line width=0.1mm, black] (1,3) -- (3,1);
    \draw[dashed, line width=0.1mm, black] (1,3) -- (1,4.5);
    \draw[dashed, line width=0.1mm, black] (4.5,3.5) -- (3.5,4.5);
    
    \draw[dashed, line width=0.1mm, black] (1,4.5) -- (3.5,4.5);
    \draw[dashed, line width=0.1mm, black] (4.5,3.5) -- (4.5,1);
    \draw[dashed, line width=0.1mm, black] (4.5,1) -- (3,1);
    
    \coordinate (A) at (1,3);
    \coordinate (B) at (3,1);
    \coordinate (Q) at (4.5,3.5);
    \coordinate (R) at (3.5,4.5);
    \node at (1.2,1.2) {$ \mathcal{P}(b)$};
    \node at (4.7,4.7) {$ \mathcal{P}'(p)$};
    \node at (3,3) {$\mathcal{Q}(p, b)$};
\end{tikzpicture}
    \caption{Supermodular and submodular polyhedra, $ \mathcal{P}'(p)$ and $ \mathcal{P}(b)$, their intersection is the g-polymatroid $\mathcal{Q}(p,b)$.}
    \label{fig:sub_super_polyhedra}
\end{figure}
\begin{theorem}[Sum Theorem]\label{thm:g_polymatroid_sum}\cite[Theorem 14.2.15]{Frank2011ConnectionsOptimization}
    \newline
    The Minkowski sum of a set of g-polymatroids is given by 
    \begin{equation}
        \sum_i \mathcal{Q}(p_i, b_i) = \mathcal{Q}\left(\sum_i p_i, \sum_i b_i\right). 
    \end{equation}
\end{theorem}
\begin{remark}
    This result implies that the family of g-polymatroids forms a convex cone under Minkowski addition and non-negative scalar multiplication \cite{Edmonds2003SubmodularPolyhedra}. Adding g-polymatroids together or scaling them by a positive factor preserves the paramodularity of the super- and submodular functions that generate them, and hence results in another g-polymatroid. This property is particularly useful for aggregating the flexibility sets of multiple devices, as we will see in the following subsections.
\end{remark}

\begin{figure*}[t]
    \centering
    \begin{tikzpicture}
        \node at (0, -2) [left] {};
    \end{tikzpicture}
    \begin{subfigure}[b]{0.3\textwidth}
        \centering
            \begin{tikzpicture}

        \draw[->] (-.5,0) -- (4.5,0);
        \draw[dashed, opacity=0.4] (0,1) -- (4.5,1);
        \node at (0, 1) [left] {$\overline{x}$};
        \draw[->] (0,-2) -- (0,2.2);


        \foreach \x in {1,2,3,4}
            \draw (\x,0.1) -- (\x,-0.1);
        
        \foreach \x in {1,2,3,4}
            \draw (\x - 0.5 ,0) -- (\x-0.5,-0) node[below, font=\tiny] {\x};

        \foreach \x in {1}
            \fill[blue, opacity=0.2] (\x-1,-2) rectangle (\x,2);
        \draw[-, blue] (0, 0) -- (1, 1);\draw[-, blue] (1, 1) -- (2, 1);\draw[-, blue] (2, 1) -- (3, 1);\draw[-, blue] (3, 1) -- (4, 1);
        \draw[blue, dashed] (0, 0) -- (1, 1);\draw[blue, dashed] (1, 1) -- (4, -2);\draw[blue, dashed] (2, 1) -- (3, 1);\draw[blue, dashed] (3, 1) -- (4, 1);
            
    \end{tikzpicture}

    
        \label{fig:sub1}
        \caption*{$A=\{1\}$}
    \end{subfigure}
    \hfill
    \begin{subfigure}[b]{0.3\textwidth}
        \centering
            \begin{tikzpicture}

        \draw[->] (-.5,0) -- (4.5,0);
        \draw[dashed, opacity=0.4] (0,1) -- (4.5,1);
        \draw[->] (0,-2) -- (0,2.2);

        \foreach \x in {1,2,3,4}
            \draw (\x,0.1) -- (\x,-0.1);
        
        \foreach \x in {1,2,3,4}
            \draw (\x - 0.5 ,0) -- (\x-0.5,-0)  node[below, font=\tiny] {\x};

        \foreach \x in {1, 2}
            \fill[blue, opacity=0.2] (\x-1,-2) rectangle (\x,2);
        \draw[-, blue] (0, 0) -- (1, 1);\draw[-, blue] (1, 1) -- (2, 1);\draw[-, blue] (2, 1) -- (3, 1);\draw[-, blue] (3, 1) -- (4, 1);
        
        \draw[blue, -] (0, 0) -- (1, 1);\draw[blue, -] (1, 1) -- (2, 1);\draw[blue, dashed] (2, 1) -- (3, 0);\draw[blue, dashed] (3, 0) -- (4, -1);
            
    \end{tikzpicture}

    
        \label{fig:sub2}
        \caption*{$A=\{1,2\}$}

    \end{subfigure}
    \hfill
    \begin{subfigure}[b]{0.3\textwidth}
        \centering
            \begin{tikzpicture}

        \draw[->] (-.5,0) -- (4.5,0) node[right, font=\small] {$t$};
        \draw[dashed, opacity=0.4] (0,1) -- (4.5,1);
        \draw[->] (0,-2) -- (0,2.2);


        \foreach \x in {1,2,3,4}
            \draw (\x,0.1) -- (\x,-0.1);
        
        \foreach \x in {1,2,3,4}
            \draw (\x - 0.5 ,0) -- (\x-0.5,-0) node[below, font=\tiny] {\x};

        \foreach \x in {1, 3}
            \fill[blue, opacity=0.2] (\x-1,-2) rectangle (\x,2);
        \draw[blue, -] (0, 0) -- (1, 1);\draw[blue] (1, 1) -- (2, 1);\draw[blue, -] (2, 1) -- (3, 2);\draw[blue] (3, 2) -- (4, 2);\node at (4, 2) [right, font=\small] {$b_T(A)$};
        \draw[dashed, blue] (0, 0) -- (1, 1);\draw[dashed, blue] (1, 1) -- (2, 0);\draw[dashed, blue] (2, 0) -- (3, 1);\draw[dashed, blue] (3, 1) -- (4, 0);\node at(4, 0) [above right, font=\small] {$x(t)$};

    \end{tikzpicture}

    
        \label{fig:sub3}
        \caption*{$A=\{1,3\}$}
    \end{subfigure}

\caption{The submodular function $b^T$ for various subsets $A \subseteq \mathcal{T}$. The black dashed line represent the energy limits $\overline{x}$. The blue dashed line represents the consumption profile $x(t)$ that maximizes its consumption over time steps in $A$.
    The solid blue lines represent the cumulative energy consumption for the device during time steps in $A$ following consumption profile $x(t)$, and $b^T(A)$ is the maximum cumulative consumption during time steps in $A$ at the end of the time horizon.}
    \label{fig:paramodular}
\end{figure*}
\subsection{Individual Flexibility Sets as Generalized Polymatroids}
For clarity of notation, in this subsection we shall drop the subscript $i$ over elements of the population, and construct the flexibility sets for a device with charging parameters $\xi = (\underline{u}, \overline{u}, \underline{x}, \overline{x})$.
To leverage Theorem~\ref{thm:g_polymatroid_sum}, it is first necessary to establish that the individual flexibility sets defined in Definition~\ref{dfn:individual_flexibility_sets} are indeed g-polymatroids. 
The following new Theorem makes this explicit:

\begin{theorem}
    \label{lem:individual_flexibility_sets_g_polymatroid}
    $\mathcal{F}(\xi)$ is the g-polymatroid $\mathcal{Q}(p_T, b_T)$, where
    $p_0(A) := \underline{u}(A)$ and $b_0(A) := \overline{u}(A)$ and 
\begin{subequations}\label{eq:individual_param}
    \begin{IEEEeqnarray}{rCl}
        p_{s+1}(A) & = & \max\Bigl\{ p_s(A \cap S), \;\;\underline{x}(s+1) - b_s(A'\cap S)\Bigr\}
    \nonumber\\
    & & \hphantom{\max\Bigl\{} + \underline{u}(A \cap S')
    \label{eq:individual_super}
    \\[2pt]
    b_{s+1}(A) & = & \min\Bigl\{ b_s(A \cap S), \;\;\overline{x}(s+1) - p_s(A'\cap S)\Bigr\}
    \nonumber\\
    & & \hphantom{\min\Bigl\{} + \overline{u}(A \cap S')
    \label{eq:individual_sub}
    \end{IEEEeqnarray}
\end{subequations}
    for $s = 0,\dots,T-1$, where $S:=\{1,...,s+1\}$.
\end{theorem}
This result is formally proved in the Appendix.

\begin{remark}
    Theorem \ref{lem:individual_flexibility_sets_g_polymatroid} provides a characterization of the individual flexibility sets in terms of g-polymatroids. The corresponding super- and submodular functions, $p_T$ and $b_T$, admit an intuitive interpretation: for each subset $A \subseteq \mathcal{T}$, $b_T(A)$ specifies the maximum cumulative energy the device can consume during time steps within the interval $A$. Equivalently, $p_T(A)$ represents the minimum cumulative energy the device can consume (or the maximum it can generate) within those time steps.  From their derivation $p_T$ and $b_T$ implicitly enforce the device's operational constraints. We show the super- and submodular functions for a simple case in Fig. \ref{fig:paramodular}.
\end{remark}
 
\begin{remark}\label{rem:genpolyedges}
    G-polymatroids are a natural way of expressing the flexibility in DERs. The edges of a g-polymatroid are directed as multiples of the vectors $e_j$ or $e_j - e_k$, $\forall j,k \in \mathcal{T}$, \cite{Frank2014CharacterizingPolymatroids}. Equivalently, the normal fans associated with the flexibility sets are coarsenings of the (projected) braid fan, \cite[3.2]{PostnikovAlex2008FacesPermutohedra.}. Consider the state of charge dynamics from \eqref{eq:soc_dynamics}, this endows our flexibility sets with some symmetry: charging more during one time step necessitates charging less by an equivalent amount during another. For example, let $v$ be an extreme point of $\mathcal{F}(\xi)$, representing an extremal consumption profile for the device. Under \eqref{eq:soc_dynamics}, increasing the $j^{th}$ component of $u$ (to charge more at time $j$) forces a commensurate decrease in the $k^{th}$ (charging less at time $k$), effectively moving in the direction $e_j - e_k$.  
    One can view this as a discrete geometric analogue of \cite[Lemma III.1]{Evans2020AResources}. 
    This symmetry in the SoC dynamics is precisely what enables us to characterize $\mathcal{F}(\xi)$ as a g-polymatroid. For devices with leaky charging dynamics or non-perfect charging efficiency this symmetry is broken and so their flexibility sets are not g-polymatroids. Nonetheless, g-polymatroids can serve as good inner approximations of these flexibility sets \cite{Mukhi2025AggregatePolymatroids}.
    \end{remark}

\subsection{Aggregate Flexibility Sets}
    We are now in a position to apply Theorem \ref{thm:g_polymatroid_sum} to provide an exact representation of the aggregate flexibility set of a population of devices. In the following we omit the subscript $T$ on the super- and submodular functions, and reintroduce the subscript $i \in \mathcal{N}$, so that the flexibility set of device $i$ is denoted $\mathcal{F}(\xi_i) = \mathcal{Q}(p_i, b_i)$.

    \begin{theorem}[Aggregation]\label{thm:agg_flex_set}
        The aggregate flexibility set $\mathcal{F}(\Xi_{\mathcal{N}})$ is a g-polymatroid, denoted $\mathcal{F}(\Xi_{\mathcal{N}}) = \mathcal{Q}(p, b)$, where
        \begin{equation}\label{eq:agg_super_sub}
            p(A) = \sum_{i \in \mathcal{N}} p_i(A), \quad 
            b(A) = \sum_{i \in \mathcal{N}} b_i(A).
        \end{equation}
        \end{theorem}
    \begin{proof}
        The proof follows directly from applying Theorem \ref{lem:individual_flexibility_sets_g_polymatroid} to Theorem \ref{thm:g_polymatroid_sum}.
    \end{proof}
    \begin{remark}
        
    The super- and submodular functions $p$ and $b$ for the aggregate flexibility set $\mathcal{F}(\Xi_{\mathcal{N}})$ are derived as the sums of the corresponding functions for the individual flexibility sets $\mathcal{F}(\xi_i)$.
    These set functions, $p$ and $b$, exactly encode all constraints of the aggregate flexibility of the population of devices.
    Referring to Definitions \ref{dfn:submodular} and \ref{dfn:g_polymatroid}, the aggregate flexibility set $\mathcal{F}(\Xi_{\mathcal{N}})$ is defined by $2^{\mathcal{T}+1}$ hyperplanes, providing a lower and upper bound for each subset of $\mathcal{T}$, as determined by the super- and submodular functions.
    For practically relevant values of $\mathcal{T}$, computing the full set of hyperplanes is computationally infeasible.
    However, as we demonstrate in the following section, it is unnecessary to compute the entire set to solve many relevant problems.
    \end{remark}

\subsection{EV Aggregation}
Theorem~\ref{lem:individual_flexibility_sets_g_polymatroid} defines super- and submodular set functions that exactly characterize the flexibility of a broad class of DERs. For particular device categories, these functions admit further simplifications, allowing the aggregate supermodular and submodular functions in Theorem~\ref{thm:agg_flex_set} to be written more compactly. In this subsection, we demonstrate this reduction for a fleet of charging-only EVs with heterogeneous arrival and departure times. While other such simplifications are also possible, we restrict our discussion to this example for brevity.

By substituting the EV‐specific parameters \(\underline{x}_i,\;\overline{x}_i,\;\underline{u}_i\) and \(\overline{u}_i\) from Section~\ref{subsect:expressivity} into 
Theorem~\ref{lem:individual_flexibility_sets_g_polymatroid}, in particular noting that $\underline{u}_i(t) = 0$ for all $t$, the super‐ and submodular functions for each device simplify to

\begin{subequations}\label{eq:simplified_ev_modular}
\begin{align}
    p_i(A) &= \max\bigl\{\underline{e}_i \;-\; \lvert C_i \setminus A\rvert\,m_i,\; 0\bigr\},\\
    b_i(A) &= \min\bigl\{\lvert A \cap C_i\rvert\,m_i,\;\overline{e}_i\bigr\}.
\end{align}
\end{subequations}
To simplify the super- and submodular functions of the aggregate flexibility set, we first introduce for each vehicle two extremal charging profiles: \(\underline{v}_i\), which delays charging as long as possible, and \(\overline{v}_i\), which front‐loads charging to finish as early as possible:
\begin{equation*}
        \underline{v}_i(t) = 
        \begin{cases}
            0                        & t <  \underline{q}_i \\
            \underline{r}_i          & t =  \underline{q}_i \\
            m_i                      & t \geq  \underline{q}_i \\
        \end{cases}
\quad
        \overline{v}_i(t) = 
        \begin{cases}
            m_i                    & t <  \overline{q}_i \\
            \overline{r}_i         & t =  \overline{q}_i \\
            0                      & t \geq \overline{q}_i \\
        \end{cases}
    \end{equation*}
where $\underline{q}_i = \lfloor \underline{e}_i / m_i \rfloor$, $\underline{r}_i = \underline{e}_i - \underline{q}_i m_i$, and similarly for $\overline{q}_i$ and $\overline{r}_i$. We can then write \eqref{eq:simplified_ev_modular} as:
\[
    p_i(A) \;=\;\sum_{t=1}^{|A \cap C_i|} \underline{v}_i(t), 
    \quad
    b_i(A) \;=\;\sum_{t=1}^{|A \cap C_i|} \overline{v}_i(t).
\]
Now we define $ \mathcal{N}_{a,d} := \{i\in\mathcal{N} : a_i = a,\ d_i = d\}$ as the subset of EVs that share the same arrival and departure time. 
Applying Theorem~\ref{thm:agg_flex_set} to this subset of devices we obtain the the aggregate super- and submodular functions for a set of EVs with homogeneous arrival and departures times
\begin{equation*}
  p_{a,d}(A) = \sum_{t=1}^{\lvert A\cap [a,d]|} \underline v_{a,d}(t),
  \quad
  b_{a,d}(A) = \sum_{t=1}^{\lvert A\cap [a,d]|} \overline v_{a,d}(t).
\end{equation*}
where
\[
  \underline v_{\,a,d}(t) = \sum_{i\in\mathcal{N}_{a,d}} \underline v_i(t),
  \quad
  \overline v_{\,a,d}(t)=\sum_{i\in\mathcal{N}_{a,d}} \overline v_i(t).
\]
These super- and submodular functions define the \textit{regular permutahedron} \cite{Postnikov2009PermutohedraBeyond}. From this, we recover the characterisation of the aggregate flexibility set as a permutahedron, as derived in \cite{Mukhi2023AnVehicles} and \cite{Panda2024EfficientVehicles}. By summing the super- and submodular functions across all arrival and departure intervals, we obtain the fleet-wide super- and submodular functions as:
\begin{equation*}\label{eq:V1G_aggregate_permutahedra}
    p(A) = \sum_{a < d}\sum_{t=1}^{\lvert A\cap [a,d]|} \underline v_{a,d}(t)
    \quad
    b(A) = \sum_{a < d}\sum_{t=1}^{\lvert A\cap [a,d]|} \overline v_{a,d}(t).
\end{equation*}
This yields a representation of the  aggregate flexibility that is independent of the number of devices in the population. As a result, even for extremely large device fleets, the aggregate flexibility can be represented compactly.

\section{Optimization}\label{sec:optimization}
With an exact representation of the aggregate flexibility of a population of DERs derived in the previous section, we now focus on optimizing over this set. From Theorem \ref{thm:agg_flex_set} we are given a representation of the aggregate flexibility set as the g-polymatroid $\mathcal{F}(\Xi_N) = \mathcal{Q}(p, b)$, defined by $2^{T+1}$ hyperplanes. For practical values of $T$ explicitly representing all constraints of $\mathcal{Q}(p, b)$ is infeasible. However, as we shall show in this section optimizing over the sets is indeed feasible.
In particular we consider how to solve general problems of the form
\begin{equation}\label{eq:general_prob}
    \begin{aligned}
        & \underset{}{\text{minimize}}
        & & f(u) \\
        & \text{subject to}
        & & u \in \mathcal{Q}(p,b),  \quad 
        Cu \leq d.
    \end{aligned}
\end{equation}
We present this formulation in a general form, as it can be used to model a broad class of optimization problems relevant to an aggregator. To solve the linear and non-linear variants of \eqref{eq:general_prob} we can apply Dantzig-Wolfe or Frank-Wolfe decomposition \cite{Dantzig1960DecompositionPrograms} \cite{Frank1956AnProgramming}. However, the computational efficiency of these methods relies on the existence of fast and scalable algorithms for solving linear programs over $\mathcal{Q}(p,b)$, i.e. solving:

\begin{equation}\label{eq:lp_g_polymatroid}
    \underset{}{\text{minimize}} \;\;  c^Tu \quad
    \text{subject to} \quad u \in \mathcal{Q}(p,b). 
\end{equation}
Conveniently, g-polymatroids provide an efficient method of solving this class of problems, so the rest of this section is concerned with solving problems of the form of \eqref{eq:lp_g_polymatroid}.

\subsection{Lifting to Base Polyhedron}
We first show how one can define $\mathcal{Q}(p,b)$ as a projection of a \textit{base polyhedron}. 

\begin{definition} The \emph{base polyhedron} associated with $b$ is the intersection of the submodular polyhedron $ \mathcal{P}(b)$ and the plane $u(\mathcal{T}) = b(\mathcal{T})$:
\begin{equation*}
    \mathcal{B}(b) := \left\{ u \in  \mathbb{R}^{\mathcal{T}} \mid u(A) \leq b(A) \;\; \forall A \subseteq \mathcal{T}, u(\mathcal{T}) = b(\mathcal{T})\right\}.
\end{equation*} 
\end{definition}
To recast $\mathcal{Q}(p,b)$ as a base polyhedra, we extend the ground set $\mathcal{T}$ with a new element $\Tilde{t}$, such that $\Tilde{\mathcal{T}} := \mathcal{T} + \Tilde{t}$. 

\begin{theorem}[Projection]\label{thm:projection}\cite[Theorem 14.2.4]{Frank2011ConnectionsOptimization} 
$\mathcal{Q}(b,p)$ is the projection of $\mathcal{B}(\Tilde{b})$ along $\Tilde{t}$, where 
\begin{equation*}
        \Tilde{b}(A) := 
        \begin{cases}
            b(A)                    & A \subseteq \mathcal{T}\\
            -p(\mathcal{T} \setminus A)     & \Tilde{t} \in A.
        \end{cases}
\end{equation*}
\end{theorem}
Alternatively put, if $\Tilde{u} \in \mathbb{R}^{\Tilde{\mathcal{T}}}$ is feasible in $\mathcal{B}(\Tilde{b})$ then its linear projection, $u \in \mathbb{R}^\mathcal{T}$, along $\Tilde{t}$, obtained by omitting $\Tilde{u}(\Tilde{t})$, will also be feasible in $\mathcal{Q}(p,b)$. 

\subsection{A Greedy Algorithm}\label{subsection:greedy}
Using Theorem \ref{thm:projection} the problem in \eqref{eq:lp_g_polymatroid} can be restated as a linear optimization over a submodular base polyhedron:
\begin{equation}\label{eq:lp_base}
    \underset{}{\text{minimize}} \;\; \Tilde{c}^T \Tilde{u}, \quad
    \text{subject to} \quad \Tilde{u} \in \mathcal{B}(\Tilde{b}), 
\end{equation}
with the cost vector extended by $\Tilde{c}(\Tilde{t}) = 0$ and $\Tilde{c}(t) = c(t) \; \forall t \in \mathcal{T}$.
A fundamental result in submodular optimization states that linear programs over a submodular base polyhedron admit a greedy solution procedure \cite{Fujishige2005SubmodularOptimization}. For completeness we outline this greedy procedure below. 

Consider the symmetric group $\mathrm{Sym}(\Tilde{\mathcal{T}})$, which is the group of all permutations of the set $\Tilde{\mathcal{T}}$. Let $\pi$ be the permutation in $\mathrm{Sym}(\Tilde{\mathcal{T}})$ that arranges the components of $\Tilde{c}$ in non-decreasing order, i.e.,
\begin{equation}\label{eq:cost_order}
    \Tilde{c}(\pi(1)) \leq \Tilde{c}(\pi(2)) \leq ... \leq \Tilde{c}(\pi(T+1)). 
\end{equation}
For $t \in \{0, 1,..., T + 1\}$, define $S_t := \{\pi(1), ..., \pi(t)\}$ as the set of the first $t$ elements of $\pi$. By definition $S_0 = \emptyset$ and $S_{\Tilde{\mathcal{T}}} = \Tilde{\mathcal{T}}$. 
Now we construct $\Tilde{u}^* \in \mathbb{R}^{\Tilde{\mathcal{T}}}$ as follows:
\begin{equation}
    \Tilde{u}^*(t) = \Tilde{b}(S_t) - \Tilde{b}(S_{t-1}) \quad \forall \; t \in \Tilde{\mathcal{T}}.
\end{equation}

\begin{theorem}[Greedy Algorithm]\label{thm:greedy_alg}\cite[Theorem 14.5.2]{Frank2011ConnectionsOptimization}
    $\Tilde{u}^*$ is in the base polyhedron $\mathcal{B}(\Tilde{b})$ and an optimal solution to \eqref{eq:lp_base}.
\end{theorem}
\begin{remark}   
Note that constructing \(\tilde u^*\) via the greedy algorithm requires exactly \(T+1\) evaluations of the submodular function \(\tilde b\), one for each marginal increment \(\tilde b(S_t)-\tilde b(S_{t-1})\).  Because these evaluations are mutually independent, they can be distributed and executed in parallel, so that the overall wall‐clock time is dominated by a single call to \(\tilde b\) plus the \(O(T\log T)\) cost of sorting the entries of \(\tilde c\).  Once \(\tilde u^*\) is obtained, the optimal decision \(u^*\) for the original g-polymatroid LP \eqref{eq:lp_g_polymatroid} follows immediately by discarding the auxiliary component \(\tilde u^*(\tilde t)\).  Embedding this within Dantzig–Wolfe or Frank–Wolfe decomposition schemes then enables efficient solving of both linear and nonlinear variants of the general aggregator problems of \eqref{eq:general_prob}, even for large time horizons.  
\end{remark}

We conclude this section by describing how to label the vertices of $\mathcal{B}(\Tilde{b})$, and by extension $\mathcal{Q}(p,b)$. This labeling will become relevant for the disaggregation process described in the following section. In the greedy algorithm, we select the permutation $\pi \in \mathrm{Sym}(\Tilde{\mathcal{T}})$ that arranges elements of $\Tilde{c}$ in non-decreasing order. This particular permutation uniquely determines the vertex of $\mathcal{B}(\Tilde{b})$ that correspond to the optimal solution of \eqref{eq:lp_base}. Consequently, there is a surjection between elements of  $\mathrm{Sym}(\Tilde{\mathcal{T}})$ and the vertices of $\mathcal{B}(\Tilde{b})$, and hence $\mathcal{Q}(p,b)$. Therefore, we can label the vertices of $\mathcal{Q}(p,b)$, denoted $ \mathcal{V}_{\mathcal{Q}(p,b)}$, with elements of $\mathrm{Sym}(\Tilde{\mathcal{T}})$: 
\begin{equation}\label{eq:vertex_set}
    \mathcal{V}_{\mathcal{Q}(p,b)}  = \left\{ v^\pi | \pi \in \mathrm{Sym}(\Tilde{\mathcal{T}}) \right\}.
\end{equation}
This is possible because the normal fan of submodular base polyhedra is the braid fan. This statement is related to the comments made in Remark \ref{rem:genpolyedges}. A detailed discussion of this is beyond the scope of this paper; however, interested readers are referred to \cite{Postnikov2009PermutohedraBeyond} for an exposition of this.

\section{Disaggregation}\label{sec:disaggregation}

With an optimal aggregate consumption profile, $u^*_\mathcal{N}$, given by the solution to \eqref{eq:general_prob}, the final task of an aggregator is to disaggregate this among devices in the population. This involves finding a feasible consumption profile for each device, such that 
the sum of the profiles is equal to the optimal aggregate consumption profile, as formalized in \eqref{eq:disaggregation}. To provide some intuition behind the disaggregation process we propose, we draw on the following two results.
\begin{theorem}[Carathéodory’s theorem]\cite[Proposition 1.15]{Ziegler2012LecturesPolytopes}\label{thm:caratheodory}
Let \(\mathcal{F}\subset\mathbb{R}^{\mathcal{T}}\) be a polytope and \(u\in\mathcal{F}\), where $|\mathcal{T}| = T$.  Then there exist \(T+1\) vertices \(v^1,\dots,v^{T+1}\in\mathcal{V}_{\mathcal{F}}\) and non-negative weights \(\lambda_1,\dots,\lambda_{T+1}\) with \(\sum_{j=1}^{T+1}\lambda^j=1\) such that $u \;=\;\sum_{j=1}^{T+1}\lambda^j\,v^j$.
\end{theorem}
\noindent
This provides us with a method of expressing any point in $\mathcal{F}_\mathcal{N}$ in terms of, at most, $T+1$ of its vertices.
\begin{corollary}[Vertex Decomposition]\cite[Corollary 2.2]{Fukuda2004FromPolytopes}\label{thm:fukuda}
Let \(\mathcal{F}_{\mathcal{N}}=\sum_{i\in \mathcal{N}}\mathcal{F}_i\).  A point \(v\in\mathcal{F}_{\mathcal{N}}\) is a vertex of \(\mathcal{F}_{\mathcal{N}}\) if and only if there exists a permutation \(\pi\in\mathrm{Sym}(\Tilde{\mathcal{T}})\) such that $ v \;=\;\sum_{i=1}^N v_i^\pi,$ where \(v_i^\pi\) is the unique vertex of \(\mathcal{F}_i\) selected by the ordering \(\pi\) (cf.\ \eqref{eq:vertex_set}).
\end{corollary}
\noindent
Essentially this means that all vertices of $\mathcal{F}_\mathcal{N}$ can be decomposed into vertices of the summands that define $\mathcal{F}_\mathcal{N}$. These two results allow us to provide a method of disaggregating $u_\mathcal{N}$ as we show in the following theorem.
\begin{theorem}[Disaggregation]\label{thm:disaggregation}
    For all $u_\mathcal{N} \in \mathcal{F}_\mathcal{N}$, there exist $\lambda \in \mathbb{R}^{T + 1}$, and $\Pi = \{\pi_1, ..., \pi_{T+1}\} \subset \mathrm{Sym}(\Tilde{\mathcal{T}})$, such that
    \begin{equation}
        u_i = \sum_{j=1}^{T+1} \lambda^j v_i^{\pi_j} \in \mathcal{F}_i \quad \textrm{and} \quad
        u_\mathcal{N} = \sum_i^N u_i.
    \end{equation}
\end{theorem}
\noindent
The existence of $\lambda$ and $\Pi$ follows from Theorem \ref{thm:caratheodory} and the decomposition among the $\mathcal{F}_i$ follows from Corollary \ref{thm:fukuda}, a formal proof is presented in the Appendix.
Solving the optimization problems from \eqref{eq:general_prob} described in the previous section will give an optimal solution in the form $u_{\mathcal{N}}^* = \sum_{j=1}^{T+1} \lambda^j v^{\pi_j}$. Using $\lambda$ and $\Pi$ we can immediately apply Theorem~\ref{thm:disaggregation} for the disaggregation process.

\section{Numerical Results}\label{sec:numerical_results}

In this section, we present numerical results to evaluate the performance of the proposed aggregation methods. We first compare their execution times against established benchmarks, and then provide a case study illustrating their practical application. The results confirm that our methods consistently outperform comparable approaches reported in the literature.

\subsubsection{Benchmarking}
\begin{figure}[t]
    \centering
    \includegraphics[width=\columnwidth]{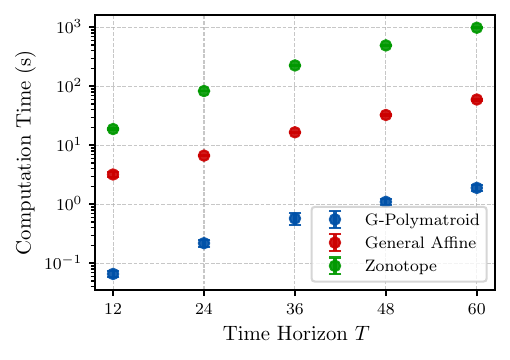}
    \caption{Computation time of various approximation methods to  solve an LP over the aggregate flexibility set of 100 DERs.}
    \label{fig:benchmarking}
  \end{figure}
We first compare the computation time of the proposed aggregation methods against existing benchmarks, considering the \textit{general affine} \cite{Taha2024AnPopulations} and \textit{zonotope-based} \cite{Muller2019AggregationResources} aggregation methods. 
We generate populations of 100 DERs and measure the time each method takes to solve the following LP:
\begin{equation}\label{prob:cost_min}
\begin{aligned}
    \textrm{minimize} \;\; &c^T u \\
    \textrm{s.t.} \;\; &u \in \mathcal{F}(\Xi_{\mathcal{N}}), \quad 
                             Cu \leq d.
\end{aligned}
\end{equation}
where $c\in \mathbb{R}^\mathcal{T}$ denotes the cost vectors of energy over the time horizon and encode network coupling constraints.
This process is repeated for various time horizon lengths, where for each horizon, we generate 100 random instances by independently sampling new populations, network constraints, and energy price profiles. 
The results are plotted in Figure~\ref{fig:benchmarking}, which demonstrate that the proposed methods achieve substantially faster computation times compared to existing approaches in the literature. These findings highlight the scalability and efficiency of the proposed framework, particularly as the problem size and time horizon increase. It is also important to note that the methods from the literature are based on inner approximations of the feasible set, and thus the solutions they produce may be sub-optimal, whereas the proposed methods retain optimality guarantees.

\subsubsection{Case Study}
Finally, we present a case study to illustrate how the proposed aggregation methods can be applied in practice. We consider a population of 50 EVs, half of which have discharging capabilities, and 100 households, each with a load, distributed generation and an ESS. Note that we can aggregate and optimize over larger populations, however the computational limits of the benchmarks we compare against restrict the comparison to the selected population size.
We consider the electricity costs minimization problem from \eqref{prob:cost_min}.
We assess the cost reduction relative to a baseline consumption profile for the population. The baseline consumption profile is defined as follows: each household minimizes their external energy consumption, i.e. minimize the $l_1$ norm of aggregate consumption of the load, distributed generation and the ESS, and the EVs adopt the consumption profile that charges them as soon as possible.
Electricity prices are sampled for each day of November 2022, for the GB system from Elexon \cite{ElexonPortal}, and we sample a new population of devices for each day. 
We compare the performance of the proposed method against the benchmarks introduced in the previous subsection.
In Fig. \ref{fig:case_study} we plot the cumulative cost for the electricity cost minimization problem over the 30 days, for the benchmarks and the baseline cost. Our methods achieve a significant relative cost reduction, compared to the benchmarks for the general affine and zonotope methods.

\begin{figure}[t]
    \centering
    \includegraphics[width=\columnwidth]{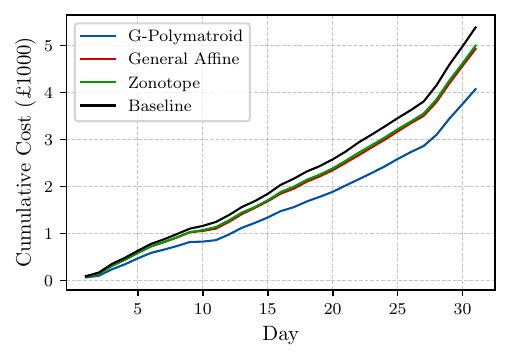}
    \caption{Cumulative electricity cost required to meet the population’s energy demand across the entire month.}
    \label{fig:case_study}
\end{figure}

\section{Conclusions}\label{sec:conc}
This paper has proposed a novel method for characterizing the aggregate flexibility of distributed energy resource populations using g-polymatroids. 
We demonstrated that, under certain assumption, the flexibility sets of individual DERs can represented as g-polymatroids, and derived the corresponding super- and submodular functions. 
Using properties of this class of polytopes we demonstrated that their aggregate flexibility can be computed efficiently.
These representations are exact providing a significant advancement over existing methods, which primarily rely on inner or outer approximations.
Due to its exactness, the proposed approach ensures both optimality and feasibility in scheduling DER flexibility.
Furthermore, we showed how the super- and submodular functions representing the aggregate flexility of a population of EVs can be simplified thereby speeding up computations.
Finally, we developed efficient optimization methods over the aggregate flexibility sets and introduced a tractable disaggregation scheme.
Our computational results confirm that this approach is viable for large-scale DER aggregations.

Future research directions include extending this framework to incorporate stochastic elements, such as uncertainties in the operation constraints of the DERs.
Additionally, incorporating network constraints into the aggregation scheme could further refine its applicability. Integrating this approach with real-time control strategies may also enhance its practical implementation in power system operations. By leveraging the properties of g-polymatroids, this work provides a theoretical framework for more efficient DER coordination in power systems.

\bibliographystyle{IEEEtran}
\bibliography{references}

\appendix
\subsection{Proof of Theorem \ref{lem:individual_flexibility_sets_g_polymatroid}}
\begin{proof}
We will show by induction on \(s=0,1,\dots,T\) that 
\begin{equation*}
    \mathcal{G}_s(\xi) := \left\{ u \in \mathbb{R}^\mathcal{T} \; \middle\vert \;
    \begin{array}{@{}cl}
        \underline{u}(t) \leq u(t) \;\; \leq \overline{u}(t) \;\; \forall t \in \mathcal{T}\\
        \; \underline{x}(t) \leq u([t]) \leq \overline{x}(t) \;\; \forall t \in [s]
    \end{array} 
    \right\}
\end{equation*}
is a g‐polymatroid \(\mathcal Q(p_s,b_s)\). Clearly by construction $\mathcal{F}(\xi) =  \mathcal{G}_T(\xi)$. 
We define the following set:
\begin{equation*}
    \mathcal{K}_{s+1}(\xi):= \left\{ u\in\mathbb{R}^\mathcal{T} \mid \underline{x}(s+1) \leq u([s+1]) \leq \overline{x}(s+1) \right\}.
\end{equation*}
The set $ \mathcal{G}_{s+1}(\xi)$ is then given by the intersection $\mathcal{G}_{s+1}(\xi) =  \mathcal{G}_s(\xi)  \cap \mathcal{K}_{s+1}(\xi)$.
We let $S \subseteq \mathcal{T}$ denote the set $S:=\{1,...,s+1\}$ and $S'$ denote its complement $S' := \mathcal{T} \setminus S$.
As $\mathcal{K}_{s+1}(\xi)$ only constrains the first $s+1$ elements of $u$, $\mathcal{K}_{s+1}(\xi)$ can be written as the Cartesian product:
\begin{equation}\label{eq:K_s+1_decomposition}
    \mathcal{K}_{s+1}(\xi) = \mathcal{K}^S_{s+1}(\xi) \times \mathbb{R}^{S'}
\end{equation}
where $\mathcal{K}^S_{s+1}(\xi) \subset \mathbb{R}^S$ is a \textit{plank}, defined as \cite[14.1]{Frank2011ConnectionsOptimization}:
\begin{equation*}
    \mathcal{K}^S_{s+1}(\xi) := \left\{ u\in\mathbb{R}^S \mid \underline{x}(s+1) \leq u([s+1]) \leq \overline{x}(s+1) \right\}.
\end{equation*}
Given a set function \( f: 2^{\mathcal{T}} \to \mathbb{R} \) and a subset \( S \subseteq \mathcal{T} \), we define the restriction of \( f \) to \( S \), denoted \( f^S \), as the function \( f^S: 2^{S} \to \mathbb{R} \) such that \( f^S(A) = f(A) \) for all \( A \subseteq S \).  
Moreover, for any subsets \( A, S \subseteq \mathcal{T} \), we define \( A_S := A \cap S \).

As the inductive hypothesis, we assume that $ \mathcal{G}_s(\xi)$ is a g-polymatroid $\mathcal{Q}(p_s, b_s)$, where the super- and submodular functions that generate it are separable amongst the disjoint subsets $S, S' \subseteq \mathcal{T}$. That is, they can be written as 
\begin{subequations}\label{eq:inductive_hype}
    \begin{equation}
        p_s(A) = p^S_s(A_S) + p^{S'}_s(A_{S'})
    \end{equation}
    \begin{equation} 
        b_s(A) = b^S_s(A_S) + b^{S'}_s(A_{S'}).
    \end{equation}
\end{subequations}
We also assume $p^{S'}_s$ and $ b^{S'}_s$ are the modular functions:
\begin{equation}\label{eq:modular_functions}
        p^{S}_s(A_{S'}) = \underline{u}(A_{S'}) \quad \text{and} \quad b^{S'}_s(A_{S'}) = \overline{u}(A_{S'}).
\end{equation}
Note by definition $ \mathcal{G}_0(\xi)$ is the g-polymatroid $\mathcal{Q}(p_0, b_0)$, where $p_0(A) = \underline{u}(A)$ and $b_0(A) = \overline{u}(A)$, which satisfies this hypothesis.
As $p_s$ and $b_s$ are separable over $S$ and $S'$, $\mathcal{G}_s(\xi)$ can be written as the Cartesian product:
$\mathcal{G}_s(\xi) = \mathcal{Q}(p^S_s,b^S_s) \times \mathcal{Q}(p^{S'}_s,b^{S'}_s)$.
Taking the intersection of $ \mathcal{G}_s(\xi)$ and $\mathcal{K}_{s+1}(\xi)$ using this and the decomposition in \eqref{eq:K_s+1_decomposition} we get
\begin{align*}
     \mathcal{G}_{s+1}(\xi)  &=   \mathcal{G}_s(\xi)  \cap \mathcal{K}_{s+1}(\xi)\\
                                &= \mathcal{Q}(p^S_s,b^S_s) \times \mathcal{Q}(p^{S'}_s,b^{S'}_s) \cap \mathcal{K}^S_{s+1}(\xi) \times \mathbb{R}^{S'}\\
                            &= \mathcal{Q}(p^S_s,b^S_s) \cap \mathcal{K}^S_{s+1}(\xi) \times\mathcal{Q}(p^{S'}_s,b^{S'}_s)  \cap \mathbb{R}^{S'}.
\end{align*}
The intersection of $\mathcal{Q}(p^{S'}_s,b^{S'}_s)  \cap \mathbb{R}^{S'}$ is trivially $\mathcal{Q}(p^{S'}_s,b^{S'}_s)$.
The intersection of $\mathcal{Q}(p^S_s,b^S_s)$ and the plank $\mathcal{K}^S_{s+1}(\xi)$ is the g-polymatroid $ \mathcal{Q}(p^S_{s+1},b^S_{s+1}) \subset \mathbb{R}^S$, where $p^S_{s+1}$ and $b^S_{s+1}$ are given by \cite[Theorem 14.2.14]{Frank2011ConnectionsOptimization}:
\begin{subequations}\label{eq:intersection_plank}
    \begin{equation}
                p^S_{s+1}(A_S) = \max\{p^S_s(A_S), \;\underline{x}(s+1) - b^S_s(A'_S)\}
    \end{equation}
                \begin{equation}
                b^S_{s+1}(A_S) = \min\{b^S_s(A_S), \;\overline{x}(s+1) - p^S_s(A'_S)\},
    \end{equation}
\end{subequations}
    where $A'_S = S \setminus A_S$.
Taking the Cartesian product of the two g-polymatroids we get the g-polymatroid 
\begin{align*}
     \mathcal{G}_{s+1}(\xi) &= \mathcal{Q}(p^S_{s+1},b^S_{s+1}) \times \mathcal{Q}(p^{S'}_s,b^{S'}_s)\\
    &= \mathcal{Q}(p_{s+1},b_{s+1}),
\end{align*}
where, using the assumption that $p^{S'}_s(A_{S'})$ and $b^{S'}_s(A_{S'})$ are the modular functions from \eqref{eq:modular_functions}, we have:
\begin{align*}
    p_{s+1}(A) &=  p^S_{s+1}(A_S) + \underline{u}(A_{S'}) \\
    b_{s+1}(A) &=  b^S_{s+1}(A_S) + \overline{u}(A_{S'}).
\end{align*}
Note, \( p_{s+1}(A) \) and \( b_{s+1}(A) \) are separable over the disjoint subsets \( S_+ := \{1, \ldots, s+2\} \) and \( S_+' := \mathcal{T} \setminus S_+ \). Moreover, \( p^{S_+'}_{s+1} \) and \( b^{S_+'}_{s+1} \) are precisely the modular functions \( \underline{u}(A_{S_+'}) \) and \( \overline{u}(A_{S_+'}) \), respectively.
Hence, if \eqref{eq:inductive_hype} and \eqref{eq:modular_functions} hold for $s$, they hold for $s+1$, thus completing the inductive step.

Finally, writing $p_{s+1}$ and $b_{s+1}$ out in full we get:
\begin{subequations}
    \begin{IEEEeqnarray}{rCl}
        p_{s+1}(A) & = & \max\Bigl\{ p_s(A \cap S), \;\;\underline{x}(s+1) - b_s(A'\cap S)\Bigr\}
    \nonumber\\
    & & \hphantom{\max\Bigl\{} + \underline{u}(A \cap S')
    \\[2pt]
    b_{s+1}(A) & = & \min\Bigl\{ b_s(A \cap S), \;\;\overline{x}(s+1) - p_s(A'\cap S)\Bigr\}
    \nonumber\\
    & & \hphantom{\min\Bigl\{} + \overline{u}(A \cap S')
    \end{IEEEeqnarray}
    \end{subequations}
as required.
\end{proof}

\subsection{Proof of Theorem \ref{thm:disaggregation}}
\begin{proof}
    By Theorem \ref{thm:caratheodory}, for all $u_\mathcal{N} \in \mathcal{F}_\mathcal{N}$, there exist $\lambda \in \mathbb{R}^{T + 1}$ and $\Pi = \{\pi_1, ..., \pi_{T+1}\} \subset \mathrm{Sym}(\Tilde{\mathcal{T}})$, such that $
        u_\mathcal{N} = \sum_{j=1}^{|\Tilde{\mathcal{T}}|} \lambda^j v_\mathcal{N}^{\pi_j}$,
    where $\sum_j^{T+1}\lambda^j = 1$, $\lambda^j > 0 \; \forall \;j$, and  $v_\mathcal{N}^{\pi_j}$ are vertices of $\mathcal{F}_\mathcal{N}$.
    Corollary \ref{thm:fukuda} allows us to write $ v_\mathcal{N}^{\pi_j}$ as a decomposition of the vertices of the $\mathcal{F}_i$, that is $v_\mathcal{N}^{\pi_j} = \sum_{i\in \mathcal{N}} v_i^{\pi_j} \;\;\forall \pi_j \in \Pi$  where $v_i^{\pi_j}$ are vertices of $\mathcal{F}_i$. Therefore, we can rwrite $u_\mathcal{N}$ as
    \begin{equation*}
        u_\mathcal{N} = \sum_{j=1}^{T+1} \lambda^j \sum_{i\in \mathcal{N}} v_i^{\pi_j} = \sum_{i\in \mathcal{N}} \sum_{j=1}^{T+1} \lambda^j v_i^{\pi_j} = \sum_{i\in \mathcal{N}} u_i.
    \end{equation*}
     where we define $u_i := \sum_{j=1}^{T+1} \lambda^j v_i^{\pi_j}$. By definition, $u_i$ is a convex combination of the vertices of $\mathcal{F}_i$ hence $u_i \in \mathcal{F}_i$, completing our proof.
\end{proof}

\vfill
\end{document}